\documentclass{article}
\usepackage{amssymb}
\usepackage{mathrsfs}
\usepackage{dsfont}
\newcommand{\concat}{{+}\!\!{+}}
\usepackage{amsthm}
\usepackage{amsmath}
\usepackage{url}
\usepackage{cite}
\usepackage{vmargin}
\usepackage{palatino}
\usepackage{mathpazo}

\newcommand{\stype}{\mathfrak{S}}
\newcommand{\stypet}{$\mathfrak{S}$}

\allowdisplaybreaks

\usepackage{graphicx}

\usepackage{authblk}

\newtheorem{theorem}{Theorem}
\newtheorem{proposition}{Proposition}
\newtheorem{lemma}{Lemma}
\newtheorem{definition}{Definition}
\newtheorem{remark}{Remark}

\AtBeginDocument{%
  \providecommand\BibTeX{{%
    \normalfont B\kern-0.5em{\scshape i\kern-0.25em b}\kern-0.8em\TeX}}}

\begin{document}

\title{A Constructive, Type-Theoretic Approach to \\ Regression via Global Optimisation}

\author{Dan R. Ghica}
\author{Todd Waugh Ambridge}

\affil{School of Computer Science, University of Birmingham, UK}
\date{}
\maketitle

\begin{abstract}
  We examine the connections between deterministic, complete, and general global optimisation of continuous functions and a general concept of regression from the perspective of constructive type theory via the concept of `searchability'. We see how the property of convergence of global optimisation is a straightforward consequence of searchability. The abstract setting allows us to generalise searchability and continuity to higher-order functions, so that we can formulate novel convergence criteria for regression, derived from the convergence of global optimisation. All the theory and the motivating examples are fully formalised in the proof assistant \textsc{Agda}. 
\end{abstract}

\section{Introduction}

For some given \emph{objective function} $f$ and set of equalities, inequalities, or arbitrary constraints $S$, the central goal of \emph{global optimisation} is to compute, with mathematical guarantees, the global minimum of $f$ subject to $S$. Global optimisation has numerous obvious applications in all areas of engineering and computational sciences, as it gives a general recipe for solving problems of arbitrary complexity. As an area of research, the study of global optimisation algorithms is mature, with a recent survey indicating more than twenty textbooks and research monographs in the last few decades~\cite{floudas2009review}.

Global optimisation algorithms fall under several categories, but in this paper we will focus on algorithmis that are:
\begin{description}
\item[General:] Algorithms may take into account information about the shape of the function. For example, the minimisation of functions with convex envelopes is intensively studied~\cite{tawarmalani2001semidefinite}. In contrast, we will make minimal assumptions of this nature.
\item[Complete:] An incomplete algorithm makes no guarantees regarding the quality of the solution it arrives at, focussing on efficiency via sophisticated heuristics, rather than correctness. The typical example of an incomplete algorithm is \emph{gradient descent}, which will only find a local minimum of a function~\cite{ruder2016overview}. In contrast, we provide mathematical guarantees a solution is indeed optimal within some margin of error.
\item[Deterministic:] A randomized algorithm can offer an asymptotic guarantee that the optimum is reached, with probability one, without actually knowing when it has been reached~\cite{SolisWets81}. In contrast, we will give strong termination guarantees for the algorithm.
\item[Continuous:] Many global optimisation algorithms deal with discrete problems, such as branch-and-bound~\cite{lawler1966branch}. In contrast, we will focus on the minimisation of continuous functions. 
\end{description}
To summarise, in this paper we will concentrate on \emph{general, complete, continuous, deterministic global search}, which finds one guaranteed optimal-within-epsilon global minimum of a continuous function~\cite{neumaier2004complete}. In the sequel, by `optimisation' this is what we mean, precisely. 

The first important results in the area relevant to our work appeared in the 1960s and 70s, the optimisation of rational functions using interval arithmetic by Moore and Young~\cite{mooreyoung59}, which was then generalised to Lipschitz-continuous functions by Piyavskii~\cite{piyavskii1972algorithm}. The idea of the algorithm is rather simple. By splitting the domain of the function into intervals, we impose a certain degree of precision on the horisontal axis. The Lipschitz constant will then bound the growth of the function on each interval, thus allowing us to calculate a precision on the vertical axis. In effect, we can `discretise' the function with known precision along both imput and output, which will make the problem decidable. It also allows the application of efficient discrete algorithms such as branch-and-bound for continuous optimisation~\cite{skelboe1974computation}.

One of the important and immediate applications of optimisation is \emph{regression}, broadly construed. It means finding some parameters for a model so that a target error (loss) function is minimised. This connection is so intuitive and obvious that it is rather surprising that it is not expressed more emphatically in the literature. This broad formulation of regression captures not just conventional regression problems (linear regression, polynomial regression, etc.) but virtually all machine learning algorithms that are sometimes referred to as `curve fitting'~\cite{pearl2018build}. 

\paragraph{A new approach: searchable types}

The inspiration for our new approach to global search and regression is in earlier work on ‘searchability’~\cite{DBLP:conf/lics/Escardo07,DBLP:journals/lmcs/Escardo08}, concerning the construction of algorithms (‘selection functions’) for finding elements in compact spaces satisfying a (computable) predicate. 
Finite sets are trivially compact, and so are trivially searchable. However, certain infinite sets are also searchable by Tychonoff's theorem, which states that the product space of any set of compact spaces is itself compact.
The infinite product of a set $X$ is given by the function space $\mathbb N\to X$, whose elements are infinitary sequences of elements of $X$. These infinitary sequences are therefore, in a certain sense, searchable; which is somewhat surprising.
This development is particularly interesting in the context of constructive real numbers, as the computable elements of compact intervals of $\mathbb R$ can be represented as infinitary sequences of digits taken from a finite set $D$.
In this work, a constructive Tychonoff-style theorem is utilised to search these representation spaces $\mathbb N \to D$ of constructive real numbers relative to certain explicit continuity conditions.

\paragraph{Contributions}

Our paper establishes new connections between several areas: global optimisation, regression, searcheable types, and constructive real numbers. This is the most important contribution. 

Our paper also makes a technical contribution to the study of searcheable types by adding an explicit requirement of continuity to the key theorems which, which allows us to formulate our key proofs in a way that is compatible with proof assistants, namely \textsc{Agda}, based on constructive type theory. This means that the entirety of our proofs are fully formalised. 

Another significant contribution is a more general methodological perspective on global optimisation and especially regression. In fact, the bulk of our paper is spent on regression, as formulated in our type-theoretic framework for searchability. The advantage of the type theoretic framework is that we can generalise the formulation of convergence of global search from $\mathbb R^n$ to the more general concept of searching on $\stype$-types, our own version of searchable types. 

Our first result is straightforward (Thm.~\ref{thm:min}); that regression can be formulated as a global minimisation property, which has a deterministic, optimal-within-epsilon solution. However, we note that this is not an actual convergence property for regression, in the sense that the Weierstrass theorem follows by interpolation (see~\cite{PINKUS20001} for an informal survey of this issue). Regression, unlike interpolation, relies on a prior assumption for the model which, if wrong, will not converge no matter how precisely we calculate its parameters. So Thm.~\ref{thm:min} only states that a solution converges on a `best guess'. 

Thus, what we give is a theorem which states, in a general setting, what it means for a regression algorithm to converge \emph{absolutely}. We distinguish between `perfect' models, which are the same as the function we aim to model (the `oracle'), provided some parameters are given the right values, and `imperfect' models in which that is not the case. One of the challenges here is to formulate the right notions of approximation between models, not just between parameters. The requisite functions, namely a loss function between models and a distortion function from models to models, are higher order. The abstract type-theoretic setting is essential here in formulating the right notions of continuity which make the theorems true. 

The most general version of convergence are Thm.~\ref{thm:congstyle} and Thm.~\ref{thm:imperfect}, which characterise the convergence of regression for an `imperfect' model. Informally, the former says that for whenever the imperfect model and the oracle are `approximately equal' the parameters of the model can be computed so that the error between the model and the oracle is approximately the same as the error introduced by the distortion function. The latter says that if the loss between an distorted oracle and the oracle is less than some $\epsilon$ then so is the loss between the regressed model and the distorted oracle. Both theorems capture the same idea: the error introduced by a `bad guess' of a model bounds the error between the regressed model and the oracle.  As an immediate consequence (Thm.~\ref{thm:perfect}) if the model is perfect (i.e. the distortion function is the identity) then the loss between the oracle and the model converges on zero.

We give some examples, mainly to show that the definitions we provide ($\stype$-types and continuity) can accommodate standard examples.

The framework that we have built for this perspective is formalised in the \textsc{Agda} programming language, which allows us to give computable (but practically inefficient) algorithms for our version of optimisation.

\section{Technical preliminaries}

\subsection{Formal proofs}

To maintain a high assurance for correctness all our main results and most of our examples are proved formally using \textsc{Agda}~\cite{DBLP:conf/tphol/BoveDN09}.
The proofs can be found online\footnote{\url{https://github.com/tnttodda/RegressionInTypeArxiv}}. We use certain options to ensure a high standard of consistency and compatibility. The `safe' option of \textsc{Agda}  disables features that may lead to possible inconsistencies, such as type-in-type or experimental or exotic constructs. This option also prevents the local disabling of termination checking. It is our explicit requirement of continuity conditions that allows all proofs to go through without violating termination, unlike prior proofs in the literature~\cite{CompactTypes}.
We also turn off the K axiom to ensure compatibility with type theories that are incompatible with `uniqueness of identity' proofs, such as homotopy type theory. Finally, using the `exact split' clause we force the type-checker require that all clauses in a definition hold as definitional equalities.  Our proofs requires several basic types and related properties found in Escard\'o's \textit{TypeTopology} library\footnote{\url{https://github.com/martinescardo/TypeTopology}}.

The bulk of the proofs of this section are in the \texttt{SearchableTypes} module, which contains annotations cross-referenced against this text. To make the presentation accessible to readers without a background in \textsc{Agda} the mathematical statements in our paper are formulated in a conventional, informal yet rigorous, mathematical vernacular. To aid the readers who are interested in formal proof details each mathematical proof is labelled with the \textsc{Agda} function formalising it.

\subsection{\stypet-types}
\label{sec:tech}

This section concerns the definition and properties of `$\stype$-types', which are used to develop the concept of Escard\'o's \textit{searchable types}. These types define the spaces in which regression can take place.

\begin{definition}[\texttt{SearchableTypes.ST-Type}]
An $\stype$-type is defined inductively as a finite non-empty type, the product of two $\stype$-types $S\times S'$, or the type of functions $\mathbb N\to S$, where $S$ is an $\stype$-type.
\end{definition}

The key technical challenge of our approach is to define a notion of \textit{(uniform) continuity} for $\stype$-types, where continuity of a function is broadly understood as `finite amounts of output only require finite amounts of input'. In this context, whenever we deal with infinite data the \textit{precision} of our observation comes into play. In the case of \stypet-types, infinite data comes from the type with shape $\mathbb N\to S$. It is natural to think of such data as sequences, which leads to a natural notion of precision-up-to-$m$ as observing the first $m$ elements of the sequence. This notion of equivalence induces the usual ultrametric on such sequences, from which we can derive a reliable definition of uniform continuity.

We generalise this intuitive notion of precision to \stypet-types as follows. First, the way we measure precision depends on the type at which we measure it; we call the type of precisions for a given type its \textit{exactness type} (the elements of this type are \textit{precisions}).
For finite data we do not afford degrees of precision, so its exactness type is the unit type. For product types we take the product of the two exactness types point-wise. Finally, for functions from $\mathbb N$ we record a natural number, which is the precision at that level, paired with the precision in the domain of the function.

\begin{definition}[\texttt{SearchableTypes.ST-Moduli}]\label{def:typemod}
The \emph{exactness type} of an $\mathfrak S$-type is defined inductively as:
\begin{itemize}
\item The exactness type of a finite set is the unit type.
\item The exactness type of a finite product of $\mathfrak S$-types is the product of their exactness types.
\item The exactness type of a function ${\mathbb N}\to S$, where $S$ is a $\mathfrak S$-type, is the product of $\mathbb N$ with the exactness type of $S$.
\end{itemize}
\end{definition}

Precision as defined above can be used to qualify equality between elements of \stypet-types. For finite data equality is not qualified by precision, and for products it is taken component-wise. For sequences $\mathbb N\to S$, equality with precision $(n, p)$, where $p$ has the exactness type of $S$, is interpreted as observing only the first $n$ elements of the sequence, with each element observed up to precision~$p$.

\begin{definition}[\texttt{SearchableTypes.ST-$\approx$}] \leavevmode
\begin{itemize}
\item Two elements of a finite set are said to be \textit{equal with precision} $p$ just if they are equal, for any $p$.
\item Two elements of a product of $\mathfrak S$-types are \textit{equal with precision} $(p_1, p_2)$, if their $i$th projections are equal with precision~$p_i$.
\item Two elements of ${\mathbb N}\to S$, with $S$ an $\mathfrak S$-type, are \textit{equal with precision }$(m, p)$ if all elements in the $m$-size prefixes are equal with precision~$p$.
\end{itemize}
\end{definition}

Note that in the definition above the types of the precision depends on the $\mathfrak S$-type as spelled out in Def.~\ref{def:typemod}. If $x,y$ are equal with precision $p$, we write $x\equiv_p y$.

The concept of `equality with precision $p$' can be adapted to predicates as \textit{logical equivalence  with precision $p$} ($\Leftrightarrow_p$) in the obvious way (formally, \texttt{SearchableTypes.ST-$\approx_p$}).

The following properties are immediate.

\begin{proposition}[\texttt{SearchableTypes.ST-$\approx$-EquivRel}]\label{prop:equivrel}
Equality with precision $p$ is an equivalence relation.
\end{proposition}
So, immediately, equality implies equality with precision $p$, for any $p$.

We are now in a position to introduce continuity for \textit{predicates} on \stypet-types. A predicate is said to be continuous if its argument only needs to be examined up to precision $p$ in order to yield an answer. Obviously, the type of $p$ is the exactness type of the argument. 

\begin{definition}[\texttt{SearchableTypes.continuous}]
  We say that a predicate $Q$ on an \stypet-type $S$ is \textit{continuous} if there exists a precision $q$ in the exactness type of $S$, such that for all $x,x':S$, whenever $x\equiv_q x'$ we also have $Q(x)\Rightarrow Q(x')$.

  We call $q$ the \emph{modulus of continuity} (MoC) of $Q$. 
\end{definition}

The same intuition applies to functions. 
\begin{definition}[\texttt{SearchableTypes.continuous$^2$}]\label{def:funccont}
  A function $f:S\to S'$, with $S, S'$ \stypet-types, is said to be \textit{continuous} if for any precision $p$ in the exactness type of $S$ there exists a precision $q$ in the exactness type of $S'$ such that for all $x,x':S$ if $x\equiv_q x'$ then $f(x)\equiv_p f(x')$.

  We call $q$ the MoC of $f$ for $p$. 
\end{definition}

\newcommand{\otype}{$\mathfrak O$}
Note that types of shape $S\to S'$ with $S, S'$ \stypet-types are not themselves \stypet-types, but they are an important class of types, which we shall call \otype-types (\textit{oracle types}, usually ranged over by the variable~$Y$).

Certain helpful properties of continuity are immediate:
\begin{proposition}[\texttt{SearchableTypes.all-$\mathbb{F}$-preds-continuous}]
All predicates and functions on finite types are continuous.
\end{proposition}
\begin{proposition}[\texttt{SearchableTypes.$\circ$-continuous}]\label{thm:composition}
If $f:S\to S'$ and $g:S'\to S''$ are continuous then so is $g\circ f:S\to S''$.
\end{proposition}

We are now ready to introduce the concept of \textit{searchability}.

A predicate is said to be \textit{detachable} if it is always decidable, i.e. either it or its negation holds.
Note that a detachable predicate is essentially a function to a two-element type, i.e. Booleans.

\newcommand{\searcher}{$\mathscr E$}
\newcommand{\searcherm}{\mathscr E}

\begin{definition}[\texttt{SearchableTypes.searcher}]
A searcher \searcher\ on an \stypet-type $S$ is a function which given a detachable and continuous predicate on $S$ returns a \textit{witness} element of $S$, for which the predicate holds if such an element exists.
\end{definition}
Since the searcher is a well-defined function, it will always return an element of $S$ even if a witness, i.e. an element satisfying the predicate, does not exist. In that case the searcher will just return some arbitrary element of $S$.

\begin{remark}
  In the \textsc{Agda} code the definition above has two parts, also involving \texttt{SearchableTypes.\\search-condition}, which spells out what it means for a witness to satisfy the predicate.
\end{remark}

We will usually denote a searcher by  \searcher.

\begin{definition}[\texttt{SearchableTypes.continuous-searcher}]
A searcher \searcher\ on $S$ is said to be \textit{continuous} if whenever given  predicates which are equivalent with precision $p$, $P\Leftrightarrow_p Q$ it returns witnesses which are equal with precision $p$, $\searcherm (P)\cong_p\searcherm(Q)$.
\end{definition}

An \stypet-type is said to be \textit{continuously searchable} if any continuous and detachable predicate on it has a continuous searcher.

We are now building towards the main theorem of this section, that all \stypet-types are in fact continuously searchable.

\begin{lemma}[\texttt{SearchableTypes.finite-ST-searchable}]\label{lem:fintych}
	All finite non-empty types are continuously searchable.
\end{lemma}
\begin{proof}
In the case of finite (non-empty) types we use induction on the size of the type. For singletons the proof is immediate, with the searcher always returning the unique element. The continuity of this searcher and the fact that it is a proper searcher are immediate. In the inductive case, given a searcher $\mathscr E_{n}$ for a set of size $n$ and some predicate $Q$ we construct a new searcher for the finite type $Fin_{n+1}=Fin_n + \{*\}$ which behaves like the old searcher if it finds a $Q$ witness and returns the additional element $inr(*)$ otherwise:
\[
\mathscr E_{n+1}(Q)=\begin{cases}
\mathscr E_n(Q) & \text{if $Q(\mathscr E_n(Q))$} \\
inr(*) & \text{otherwise}.
\end{cases}
\]
Checking that this is a continuous searcher is laborious but routine.
\end{proof}

\begin{lemma}[\texttt{SearchableTypes.product-ST-searchable}]
The product of two continuously searchable \stypet-types is continuously searchable.
\end{lemma}\label{lem:prodtych}
\begin{proof}
In the case of the product of two $\mathfrak S$-types $S\times S'$ which are searchable, we need to construct a searcher for predicate $Q$ which returns as pair a witness $(x_\epsilon,y_\epsilon):S\times S'$. Let $\mathscr E_S$ and $\mathscr E_{S'}$ be the searchers for the two types. The witnesses are computed by:
\begin{align*}
\hat y(x) &= \mathscr E_{S'}(\lambda y. Q(x,y)) \\
x_\epsilon &= \mathscr E_S(\lambda x.Q(x, \hat y(x))\\
y_\epsilon &= \hat y(x_\epsilon).
\end{align*}
These computations are obviously continuous, and the formal proof is straightforward. Verifying that these values satisfy the conditions of a correct searcher is laborious but routine.
\end{proof}
\begin{remark}
In the previous two theorems the details of checking that the defined searchers meet the required conditions are intricate, but they are also routine in a way such that our proof-assistant (\textsc{Agda}) can also make the task easy. Because of this, our reliance on a proof assistant is not onerous but, in fact, beneficial, improving the productivity of the mathematics.
\end{remark}

The previous two lemmas are perhaps unsurprising, since finite types and binary products can be searched exhaustively and component-wise, respectively. The surprising fact is that the type of infinitary sequences satisfies the same property.

Before we proceed to the main result, we note that

\begin{lemma}[\texttt{SearchableTypes.tail-decrease-mod}]
  \label{lem:cons-decrease-modulus2}
  For any natural number $n$, if a predicate $P(\alpha)$ over $S$-sequences, with $S$ an \stypet-type, has modulus of continuity $(n+1,p)$ then predicate $P(x::\alpha)$ has modulus of continuity $(n, p)$, for any $x:S$.
\end{lemma}

\begin{lemma}[\texttt{SearchableTypes.tychonoff}]\label{lem:seqtych}
Sequences of continuously searchable \stypet-types are continuously searchable.
\end{lemma}
\begin{proof}
In this case, we need to construct a searcher $\mathscr E$ for predicate Q which returns a witness $\alpha_\epsilon : \mathbb N \to S$. Let $\mathscr E_S$ be the searcher for $S$. We proceed by induction on the first projection of the modulus of continuity of $Q$, $n:\mathbb N$: 
  
  For $n=0$, we can return any element as it will vacuously satisfy the predicate. For example, $\alpha_\epsilon = \lambda n. \mathscr E_S(\lambda x.{\mathds{1}})$.
  
  For the inductive step we construct the witness like so:
\begin{align*}
\hat x(\alpha) = \mathscr E_S (\lambda x.Q(x :: \mathcal \alpha)) \\
\alpha_t = \mathscr E (\lambda \alpha.Q(\hat x(\alpha) :: \alpha)) \\
x_\epsilon = \hat x(\alpha_t) \\
\alpha_\epsilon = x_\epsilon :: \alpha_t
\end{align*}
  $\alpha_t$ is constructed using the inductive hypothesis: by Lem.~\ref{lem:cons-decrease-modulus2}, the first projection of the MoC of the searched predicate is one less than $n$. It is laborious but routine to show that the two predicates searched here are detachable and continuous.
  
  While the formal proof may look daunting, proving this witness satisfies the predicate is intuitively straightforward. Verifying that the constructed searcher $\mathscr E$ is continuous is somewhat complex, but follows from the continuity of $\mathscr E_S$ by induction on the modulus of continuity of the predicates involved.
\end{proof}

From Lem.~\ref{lem:fintych}-\ref{lem:seqtych} the key result of this section follows immediately:

\begin{theorem}[\texttt{SearchableTypes.all-ST-searchable}]
All \stypet-types are continuously searchable.
\end{theorem}

This theorem is a Tychonoff-style theorem since \stypet-types are closely related to compact types, and the definition of \stypet-types can be interpreted as all types that can be built from finite types using products, finite or countable. The theorem guarantees that the collection of types that can be used in regression is rich enough to cover many interesting examples.

\section{Generalised parametric regression}

In this preamble to our main technical results we give a semi-formal presentation of the key ideas to aid understanding and explain the method we are following.

Consider the most common form of regression, linear regression. It involves a `model' $M_{\vec k}: \mathbb R\to \mathbb R$ defined as $M_{\vec k}(x) = k_1 \cdot x + k_0$, with $k_0, k_1\in\mathbb R$. The regression task involves computing the parameters $\vec k = (k_0,k_1)\in \mathbb R^2$ such that a measure of loss, or error, between $m_{\vec k}$ and a data set $\Omega=\{(x_i, y_i)\in\mathbb R\mid 0\leq i < n\}$ is minimised. A common, but not unique, formula for such a loss function is the `least squares', defined as
\[
\Phi = \sum_{0\leq i< n}\bigl(y_i-M_{\vec k}(x_i)\bigr)^2
\]
This is essentially an optimisation problem: finding the best $\vec k\in\mathbb R^2$ to minimise the function above.

Note that the regression problem has an identical formulation for polynomial regression, where the model is a polynomial of some fixed rank $M_{\vec k}(x)=\sum_{0\leq i\leq n}k_i\cdot x^i$, except that the problem now is finding some $\vec k\in \mathbb R^n$. We work towards generalising the concepts, offering the following informal definitions first:
\begin{definition}\label{def:genreg}
We say that an \textit{oracle} is a continuous function of type $\Omega:X_1\to X_2$ .

We say that a \textit{parameterised model} is a continuous function of type $M:X_0\to (X_1\to X_2)$ .

We define a \textit{loss function} as any continuous function of type $\Phi: (X_1\to X_2)\to (X_1\to X_2)\to [0,1]$ such that $\Phi(f,f)=0$, for any $f:(X_1\to X_2)$.
\end{definition}

These definitions are still informal in the sense that we are not saying anything yet about what the $X_i$s are.
The obvious candidates for such types are computable representations of (compact subsets of) real numbers.
However, as we shall see, any \stypet-types can be used, which leads to a generalisation of existing notions of regression.

Note that a loss function is a generalisation of a metric, dropping the requirement for it be sub-additive and even symmetric. It is convenient, without loss of generality, to normalise it to the unit interval, which will be represented as a specific \stypet-type.

For readability we may write the instantiation of a model for a given parameter as $M_k=M\,k$ and the loss function in curried form, so that the quantity to minimise is written as $\Phi(M_k, \Omega)$.

Our perspective on regression, succinctly expressed, is the following:
\\[1.5ex]
\fbox{
  \begin{minipage}{.95\linewidth}\em
    The \textbf{\emph{regression problem}} consists of finding a
    parameter $k:X_0$ such that for a given {oracle} $\Omega :
    X_1\to X_2$ and \textit{model} $M : X_0 \to (X_1 \to X_2)$, the value of the loss function $\Phi(M_k,\Omega)$ is
    minimised.
  \end{minipage}
}
\\[1.5ex]
For instance, in the case of linear regression we may (naively) take $X_0=\mathbb R^2$ and $X_1=X_2=\mathbb R$; for polynomial regression $X_0=\mathbb R^n$ for some fixed $n$ and the type of the oracle as before. The loss function, least squares (or rather a normalised version thereof), is in $(\mathbb R\to \mathbb R)\to(\mathbb R\to \mathbb R)\to[0,1]$.

However, the reals $\mathbb R$  cannot be represented as an \stypet-type. In the sequel we see how to work with computable representations of certain (compact) subsets of $\mathbb R$ which are \stypet-types and lead to interesting examples, according to our motivation discussed earlier.

\subsection{Real numbers and their representations}
\label{sec:orderedrn}

\begin{remark}\label{rem:real}
  Before we proceed we need to make some important distinctions.
  The real numbers $\mathbb R$ are a well understood mathematical concept.
  In our formal perspective we are required to work with a representation, or an encoding, of the real numbers into entities that can be defined type-theoretically.
  This leads to a foundational tension between the mathematical concepts and their formal representations.
  The most significant potential problem arises from the fact that mathematical functions operate on real numbers, whereas our functions work on \emph{encodings} of real numbers (codes).
  If a function defined in our representational domain corresponds to a genuine mathematical function it is called its \emph{realiser}.
  However, we can define functions on codes which are more `intensional' in nature than mathematical functions because they have access to the internal representation of the numbers in a way that mathematical functions do not.
  Such functions are not realisers of any genuine mathematical function.
  Yet, such functions are interesting from the point of view of computer science, data science, or machine learning insofar as we see these disciplines are intrinsically algorithmic rather than purely mathematical, thus restricted to operating on codes.
  Thus, resolving this foundational tension by ensuring that all `representational' functions are genuine realisers is not something that we are concerned with in this paper, although it is an important and well-studied topic in computable real number arithmetic~\cite{DBLP:conf/mfcs/Simpson98}.
\end{remark}

As motivated by the considerations above and our leading target examples we need to consider now real numbers. In our constructive setting we clearly need to restrict ourselves to representations of some `computable' reals. More precisely, we require representations of the reals for which our desired operations (at least comparison, addition and multiplication) can be defined and are continuous.

\newcommand{\uniti}{\mathbb U}

Real numbers are used in two ways: in the general setting, as part of defining the concept of `loss function', and in examples. Because of this distinction we can conveniently use several types which serve different purposes. For the loss function we can represented the unit interval $[0,1]$ as binary sequences $\uniti=\mathbb N\to \{0,1\}$, which is clearly an \stypet-type. For this type we can define families of strict and total order relations, each of which is detachable and continuous. Each element $r:\uniti$ is an encoding of a real number in $[0,1]$; we notate the encoding of $0$ as $0_\uniti$. The interpretation is the standard one for binary numbers: $\sum_{i\in\mathbb N}r(i)\times 2^{-(i+1)}$.

\begin{definition}[\texttt{UIOrder.$<^U$}]
For any $p:\mathbb N$, a sequence $a:\mathbb U$ is said to be \textit{less-than with precision} $p$ another sequence $b:\mathbb U$, written $a<_n b$, if there is some $k:\mathbb N$, $k < p$ such that their prefixes up to $k$ are equal and $a_k < b_k$.
\end{definition}

\begin{definition}[\texttt{UIOrder.$\leq^U$}]\label{def:leq}
For any $p:\mathbb N$, a sequence $a:\mathbb U$ is said to be \textit{less-than-or-equal-to with precision} $p$ another sequence $b:\mathbb U$, written $a\leqslant_p b$, if either $a <_p b$ or $a \equiv_p b$.
\end{definition}

It is straightforward to prove that, for any $p:\mathbb N$, $<_p$ is a strict partial order, $\leqslant_p$ is a total order and, given $a,b:\mathbb U$, these predicates are decidable and continuous.

% \begin{remark}
%   The formalisation of the ordered unit interval $\uniti$ does not have any important or difficult results. It is a collection of small propositions which nevertheless need to be spelled out as they are used in many places. The formalisation is collected in Agda module \texttt{UIOrder}.
% \end{remark}

With these considerations in place we can revisit and spell out the informal parts of Def.~\ref{def:genreg}, the general formulation of regression. To cast it in type theory, we will always take $X_i$ to be \stypet-types, and we will use the representation of the unit interval $\uniti$ as the domain of the loss function. The type of the oracle is thus some $X_1\to X_2$, which is an \otype-type.

\subsection{Continuity of the loss function}

The type of the loss function is $Y\to Y\to \uniti$, with $Y$ being \otype-types. This means that the standard definition of function continuity (Def.~\ref{def:funccont}) does not apply.
In this section we define a notion of `continuity' for loss functions.

First we introduce a notion of approximate equality for functions.

\begin{definition}[\texttt{TheoremsBase.ST-$\approx^f$}]
Two functions $f,g:S\to T$ with $S,T$ being $\mathfrak S$-types are said to be \textit{equal  with precision} $p$ in the exactness type of $T$, written $f\approx_p g$ if for all $x:S$ we have that $f(x)\equiv_p g(x)$.
\end{definition}
This is an extensional definition in which all points in the domain are evaluated, but the results are compared only with precision $m$, which needs to be of the exactness type of $T$.

With this, we can define a weaker notion of continuity for model functions. 
\begin{definition}[\texttt{TheoremsBase.continuous$^M$}]
A model function $M:S\to Y$ where  $S$ is an $\mathfrak S$-type and $Y=S'\to S''$ and \otype-type is said to be \textit{weakly continuous} if for all precisions $p$ in the exactness type of $S''$ there exists a precision  $q$ in the exactness type of $S$ such that for all $k, k':S$, if $k\equiv_q k'$ then $M_k\approx_p M_{k'}$.
\end{definition}
Note that $n$ above has the exactness type for $S''$ and $m$ for $S$, respectively.

It is straightforward to show that
\begin{lemma}[\texttt{TheoremsBase.strong$\to$weak-continuity}]
Any (model) function that is continuous is also weakly continuous.
\end{lemma}

With this, we can define (weak) continuity for the loss function.

\begin{definition}[\texttt{TheoremsBase.continuous$^L$}]
  A loss function $\Phi:Y\to Y\to \uniti$,  where $Y=S\to S'$ is an $\mathfrak O$-type, is said to be \textit{(weakly) continuous} if for any precision $p$ in the exactness type of $\uniti$, there exists a precision $q$ in the exactness type of $S'$ such that if for all $g,h:Y$ if $g\approx_q h$ then for all $f:Y$ we have that $\Phi(f,g)\equiv_p\Phi(f,h)$.

  We call $q$ the MoC of $\Phi$ for precision $p$. 
\end{definition}

The definition above can be generalised so that it is continuous in both arguments. However, only this more restricted continuity of the loss function is required by the theorems below.

\subsection{Global optimisation and the convergence of regression}

We now turn our attention to a general characterisation of algorithms for regression: in what circumstances they exist and what it means for them to be correct. The standard property of regression is that a `best guess" parameter can always be produced.
\begin{theorem}\label{thm:min}
Let  $S$ be an $\mathfrak S$-type and $Y$ be an $\mathfrak O$-type and $p$ a precision in the exactness type of $\uniti$.
For any weakly continuous model $M:S\to Y$, oracle $\Omega: Y$, and continuous loss function $\Phi:Y\to Y\to \uniti$ we can construct  a parameter $k_0:S$ such that for any $k:S$ we have that $\Phi(\Omega, M_{k_0})\leqslant_p\Phi(\Omega, M_k)$.
\end{theorem}
\begin{proof} We prove this as a corollary to the more general theorem that \textit{any} continuous function $f : S \to \uniti$ has a minimum argument $k_0 : S$ such that $\forall k:S.f k_0 \leqslant_p f k$. The corollary follows because -- due to the continuity conditions on $M$ and $\Phi$ -- the function $\lambda x.\Phi(\Omega,M_x)$ is continuous.

We use induction on the structure of $S$ as an $\mathfrak S$-type. In each case we wish to construct the \textit{argmin} for $f$ with precision $p$, notated $argmin_S(f,p) : S$.

In the finite case, we proceed by induction on the number of constructions of $S$. If $S = \mathds{1}$, the unit type with the single construction $\{\star\}$, then clearly $argmin_{\mathds{1}}(f,p) = \{\star\}$. If $S = S' + \mathds{1}$ for some \stypet-type $S'$, then we proceed by inductively computing $x_0' = argmin_{S'}(\lambda x.f(inl \ x),p)$, where $inl : S' \to S$ casts the element $x:S'$ to the corresponding element in $S$. As $x_0'$ is the argmin for $f$ in $S'$ with precision $p$, and $\{\star\}$ is the corresponding argmin in $\mathds{1}$, we simply need to decide whether $f (inl \ x_0') \leqslant_p f (inr \ *)$ or $f (inl \ x_0') \leqslant_p f (inr \ *)$ -- where $inr : \mathds{1} \to S$. This is decidable because $\leqslant_p$ is decidable and a total order by Def~\ref{def:leq}.

In the product case $S = S' \times S''$, we proceed similarly to the structure of Thm.~\ref{lem:prodtych}. We construct $(x_0,y_0)=argmin_S(f,p)$ as follows:
\begin{align*}
\hat y(x) &= argmin_{S''}(\lambda y.f(x,y),p) \\
x_0 &= argmin_{S'}(\lambda x.f(x, \hat y(x),p)\\
y_0 &= \hat y(x_0).
\end{align*}
From these inductive constructions, we have that $\forall x.f(x_0,\hat y(x_0)) \leqslant_p f(x,\hat y(x))$ and $\forall x,y. f(x,\hat y(x)) \leqslant f(x,y)$. By transitivity of $\leqslant_p$ (Def~\ref{def:leq}), therefore, $\forall x, y. f(x_0,y_0) \leqslant f(x,y_0) \leqslant f(x,y)$.

In the sequence case $S = \mathbb N \to S'$, we proceed similarly to the above and by the structure of Lemma \ref{lem:seqtych}, i.e. by induction on the first projection $n : \mathbb N$ of the MoC of $f$ at point $p$. When $n = 0$ the case is vacuous. In the inductive step, we construct $\alpha_0 = argmin_S(f)$ as follows:
\begin{align*}
\hat x(\alpha) = argmin_S' (\lambda x.f(x :: \alpha),p) \\
\alpha_t = argmin_S(\lambda \alpha.f(\hat x(\alpha) :: \alpha),p) \\
x_0 = \hat x(\alpha_t) \\
\alpha_0 = x_0 :: \alpha_t
\end{align*}
$\alpha_t$ is constructed by the inductive hypothesis on the MoC, because the MoC of $\lambda \alpha.f(x :: \alpha)$ at point $(p,\star)$ for a given value $x:S'$ will be one lower than that of $f$.
Therefore, we have that $\forall \alpha.f(\hat x(\alpha_t) :: \alpha_t) \leqslant f(\hat x(\alpha) :: \alpha)$ and $\forall x,\alpha.f(\hat x(\alpha) :: \alpha) \leqslant f(x :: \alpha)$; again, the result is obtained via the transitivity of $\leqslant_p$. An additional lemma is used to finish this case that shows the output of a continuous function $f(\alpha)$ is equal to the required precision to $f(head(\alpha) :: tail( \alpha))$, where $head(\alpha) = \alpha \ 0$ and $tail(\alpha)=\lambda n.(\alpha (n+1))$. Thus, $\forall \alpha. \alpha.f(x_0 :: \alpha_t) \leqslant_p f(head \ \alpha :: tail \ \alpha) \equiv_p f(\alpha)$.

\end{proof}

This theorem seems to give a definitive constructive, type-theoretic, characterisation of regression.
However, the computational content of the proof is, on closer inspection, not satisfactory.
We can understand that more easily by instantiating the theorem on particular types, such as $Y=\mathbb U\to\mathbb U$.
Informally speaking, the proof requires finding the argmin $k_0$ of the function $f(k)=\Phi(\Omega, M_{k})$ with some fixed precision.
The way in which $k_0$ is computed is by partitioning the interval $\mathbb U$ into a finite number of intervals computed from the precision.
The continuity condition of $f$ will allow us to compute a size of these intervals which is small enough so that their images through $f$ is smaller than the precision.
In other words, for the given precision $p$ we do not need more than a certain precision of the input.
And, since there is a finite number of partitions, we can simply examine the value of $f$ for all of them and select the one in which this value is minimal.

There are two inter-related problems here.
The first one is obvious, the algorithm that is extracted out of the proof is an always-exhaustive search of the domain, up to the desired level of precision.
The second one is more subtle and it has to do with the `stability' of the algorithm.
Suppose that there are two distinct values $k_0$ and $k_0'$ for which $f(k_0)=f(k_0')$ and which is minimal for $f$.
In this situation, as we run the argument with different precisions $p$ sometimes we may get an approximation around $k_0$ as a result and sometimes we may get one around $k_0'$.
As $p$ gets smaller the algorithm is not guaranteed to converge on either of them.

The misbehaviour is not entirely surprising considering that we are attempting to compute a function, argmin, which is known not to be computable~\cite{Troelstra}.
The reason we manage to compute anything at all is because our algorithm has access to the codes of the numbers involved, so it is a function which is not a realiser of any mathematical function (also see Remark~\ref{rem:real}).

\subsection{Regressing a perfect model}

The regression Thm.~\ref{thm:min} gives a conventional characterisation of regression, but it has certain shortcomings, as discussed.
It also does not tell the whole story. Whereas it states the situation in which the loss value can be \emph{minimised} it makes no absolute statement regarding the loss itself. We therefore desire a statement which says something about the situation in which the error can not only be minimised, but also be made vanishingly small.
In other words, a \emph{convergence theorem} guaranteeing that the regressed model is arbitrarily close, as measured by the loss function, to the oracle.

In parametric regression we are epistemologically committed to a model, we just do not know its parameters and we want to calculate them from observations. The minimisation algorithm in Thm~\ref{thm:min} is always guaranteed to produce a `best guess' in terms of minimising loss, but if our bet on a particular model is the correct one then this `best guess' should be such that the the loss can be made vanishingly small. To represent this situation, instead of taking an arbitrary oracle $\Omega$ we take an arbitrary parameter $k_0$ and create a \emph{synthetic oracle} $\Omega=M_{k_0}$. The synthetic oracle has the `same shape' as the model, therefore can be approximated with arbitrarily small loss.

For this theorem we will rely on the concept of searchability, which did not come into play in the minimisation theorem Thm~\ref{thm:min}. We will call a \emph{regerssion algorithm} a \emph{regressor}.

\begin{theorem}[\texttt{LossTheorems.perfect-theorem}]\label{thm:perfect}
  Let  $S$ be an $\mathfrak S$-type, $Y$ an \otype-type,  $p:\mathbb N$ a precision,
  $\epsilon:\uniti$ a loss value such that $0_\uniti <_p\epsilon$, and  $\Phi:Y\to Y\to\uniti$ a continuous loss function.
  
  There exists a \emph{regressor} $\mathrm{reg}: (S\to Y)\to Y\to S$ such that given an element  $k_0:S$, and weakly continuous model $M:S\to Y$, we can construct $k=\mathrm{reg}\,M\,\Omega$ such that 
  $\Phi(\Omega, M_k) <_p \epsilon$, for synthetic oracle $\Omega=M(k_0)$.
\end{theorem}
\begin{proof}
This theorem is an immediate corollary of the more general Thm.~\ref{thm:imperfect} in the next sub-section.
\end{proof}

\subsection{Regressing an imperfect model}
Thm.~\ref{thm:min} states that parametric regression eventually converges on the `best possible' solution, whereas Thm.~\ref{thm:perfect} proves that if we `guess' the model correctly then the regression converges on the `absolutely best' solution. But what if we don't guess the right model? Consider the following data which is produced by oracle $\Omega(x) = x + \mathrm{sin}\, x$ in Fig.~\ref{fig:xsinx}.

\begin{figure}
\begin{center}
\includegraphics[scale=.45]{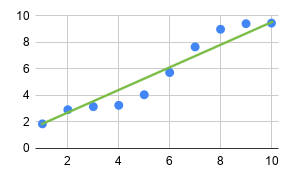}
\end{center}
\caption{Regression to imperfect model}\label{fig:xsinx}
\end{figure}
Parametric regression requires us to commit to a model, and the model can be imperfect. For instance, trying to regress a linear model $M_{\vec k}= k_1\cdot x+k_2$ for the oracle $\Omega$ could give a `pretty good' approximation, depending on the desired precision. We will aim to quantify this using another convergence theorem which essentially says that the better the guessed model the higher the precision of the approximation.

To formulate the theorem we will again use a synthetic oracle $\Omega=M_k$, with unknown parameter $k$, but we will \emph{distort} it using a function $\Psi : Y\to Y$, so that the regression will try to reconstruct $\Psi_\Omega = \Psi(\Omega)= \Psi(M_k)$ by wrongly assuming it is $\Omega$. The distortion function $\Psi$ can represent either measurement noise or a lack of perfect knowledge about the oracle.
To quantify this lack of knowledge, or how powerful the distortion is, we use two approaches.

The first theorem for regressing an unreliable model uses equality with precision $p$ to compare how `equal' the original and the distorted oracle are, and show that the loss function between the correct and the distorted oracle is `just as equal' (with precision $p$) with the loss function between the correct and the regressed oracle. It utilises the following definition of a `continuous' distortion function:

\begin{definition}[\texttt{FunEquivTheorem.continuous$^D$}]
A distortion function $\Psi : Y \to Y$, for a given $\mathfrak O$-type $Y = S \to S'$, is called \textit{continuous} if for any function $f : Y$ and precision $p$ in the exactness type of $S'$, there exists some precision $q$ in the exactness type of $S$ such that for any $x,x':S$ if $x \equiv_q x'$ then $\Psi f x \equiv_p \Psi f x'$.
\end{definition}

\begin{theorem}[\texttt{FunEquivTheorem.imperfect-corollary-with-$\approx$}]
\label{thm:congstyle}
Let  $S$ be an $\mathfrak S$-type, $Y$ an \otype-type, $p$ a precision in the exactness type of $\uniti$, and $\Phi:Y\to Y\to \uniti$ a continuous loss function.
Given an element $k_0:S$, continuous model $M: S\to Y$, and any continuous distortion function $\Psi:Y\to Y$, there exists a regressor $\mathrm{reg}: (S\to Y)\to Y\to S$ such that 
whenever $k = \mathrm{reg}\, M\, \Psi_\Omega$:
\[
\text{If }\Psi_\Omega\approx_q \Omega \text{ then }\Phi(\Omega,\Psi_\Omega)\equiv_p\Phi(\Omega,M_k),
\]
where $\Omega = M(k_0)$ is the synthetic oracle,  $\Psi_\Omega =\Psi(\Omega)$ the distorted synthetic oracle,
and $q$ is the  MoC of the loss function $\Phi$ for precision~$p$.
\end{theorem}
\begin{proof}
$S$ is an $\mathfrak{S}$-type, therefore it comes equipped with a searcher~$\mathscr E$. The regressor which computes the parameter $k$ is 
$
\mathrm{reg}\,M\,\Omega = \mathscr E(\lambda k.\Omega\approx_q M_k)
$. It turns out that, due to the searchability of the \stypet-type and the continuity conditions on the model and distortion functions, this predicate is in fact detachable and continuous.

Because there exists some $k_0:S$ such that $\Psi_\Omega \approx_q M_{k_0}$, we have by the condition on the searcher that $\Omega \approx_q M_k$ where $k = \mathrm{reg}\,M\, \Psi_\Omega$. By transitivity of $\equiv_q$ (Prop.~\ref{prop:equivrel}), we arrive at $\Omega \approx_q M_k$.

Finally, a routine calculation from the continuity of the loss function gives us the result.
\end{proof}

This theorem gives a convergence property of sorts, but it is not very useful in practice. It only applies when the distortion produced by $\Psi$ is small enough for the original and distorted oracle to be `almost equal' (with precision $q$). This means that if the distorted model differs from the true model even rarely, but by a large enough amount, the theorem does not apply. 

For this reason we also give a more practically relevant convergence theorem which uses the loss function itself to measure the degree of distortion, rather than approximate equality, and only requires a weakly continuous model. The second imperfect-model regression theorem states that if the loss between the distorted synthetic oracle and the true oracle is small, then so is the loss between the distorted synthetic oracle and the regressed model. To emphasise, this is even though the model is regressed using the distorted oracle as a source of data.

\begin{theorem}[\texttt{LossTheorems.imperfect-theorem-with-$\Phi$}]\label{thm:imperfect}
  Let  $S$ be an $\mathfrak S$-type, $Y$ an \otype-type, $p:\mathbb N$ a precision, $\epsilon :\uniti$ a loss value, and $\Phi:Y\to Y\to \uniti$ a continuous loss function.

  There exists a regressor $\mathrm {reg}:(S\to Y)\to Y\to S$ such that given an element  $k_0:S$, a weakly continuous model $M: S\to Y$, and a distortion function $\Psi:Y\to Y$,
  for parameter $k=\mathrm{reg}\,M\,\Psi_\Omega$:
  \[
  \text{if }\Phi(\Psi_\Omega, \Omega)<_p\epsilon \text{ then } \Phi(\Psi_\Omega, M_k)<_p\epsilon,
  \]
  for synthetic oracle $\Omega = M(k_0)$ and distorted synthetic oracle $\Psi_\Omega =\Psi(\Omega)$.
\end{theorem}
\begin{proof}
  The proof follows the same `recipe' as that of Thm.~\ref{thm:congstyle}, effectively constructing a regressor which has the desired property.

  The regressor will use the searcher $\mathscr E$ for the searchable type $S$ for the predicate $P(k)=\Phi(\Psi_\Omega, M_k)<_m\epsilon$ to produce the model parameter $k$. We need to show that this predicate is continuous, detachable, and satisfies the desired property, which follows from routine calculations.
\end{proof}

It is easy to see now that the perfect-model convergence theorem (Thm.~\ref{thm:perfect}) is an immediate consequence of the imperfect-model convergence theorem (Thm.~\ref{thm:imperfect}), by using the identity distortion $\Psi (\Omega) = \Omega$, which then makes $\Phi(\Psi_\Omega, \Omega)=\Phi(\Omega,\Omega)=0_\uniti<_p\epsilon$ so that the condition is trivially true.

We prefer this final formulation of the theorem, in contrast to the previous one, and we will take it as the defining property of regression, rather than the conventional minimisation one expressed in Thm.~\ref{thm:min}. 

Compared to the global minimisation approach, Thm.~\ref{thm:imperfect} has the potential to serve as a basis for more efficient algorithms.
This is because the regressor uses a searcher, which does not need to explore the search space exhaustively, unlike Thm.~\ref{thm:min}.
The searcher can stop and return the parameter as soon as the predicate is satisfied.
In other words, it will provide a `good enough', up to the specified target loss value, solution instead of searching for the `best' solution.
The `worst case' behaviour of exploring the entire space can still happen, especially if there is no witness to the predicate.

We also need to understand that the regressor is guaranteed to return a good enough parameter only when our model is a good enough guess of the oracle.
If our model is bad then the regressed parameter will not be very good either.
This is a problem in practical application, since we may not know what the true model is.
That means we cannot know whether $\Phi(\Psi_\Omega, \Omega)<_p \epsilon$.
Therefore, for the computed parameter $k$, we need to compute separately whether $\Phi(\Psi_\Omega, M_k)<_p\epsilon$.
Fortunately, the latter is computable --- that could be considered a separate `validation of regression' step.
This matches accepted practice in machine learning and data science where `learning' or `inference' is always followed by `validation' or `testing'.
What the Thm.~\ref{thm:imperfect} guarantees is that the regression algorithm is valid, in the sense that good models will always be inferred accurately. 

The imperfect-model regression theorem also saves us from relying too much on  our small methodological innovation as discussed in the Introduction. Regression as broadly practised is `from data' and not `from oracle'. In other words it is `off-line' rather than `on-line', with all data pre-sampled in advance. But we can think of off-line regression as regression to an imperfect model, with the distortion function formed by the composition of a sampling function followed by an interpolation function, noting that interpolation can be easily defined as continuous. Thm.~\ref{thm:imperfect} guarantees that if the reconstruction via sampling and interpolation is `almost perfect' then so is the regressed model. What is left unsaid is that it is indeed possible to reconstruct a function via sampling and interpolation with arbitrary precision. In other words, that the Stone–Weierstrass theorem can be recast in this setting. This is subject of further research.

\section{Examples and applications}

The framework described above is rather abstract. In this section we will show that it is applicable in a common scenario
in which regression is used: polynomial regression with a loss function in the style of least-squares. As a warm-up example we will also show a `degenerate'
form of regression, which is simply searching for the argmin of a function. This example is interesting because it gives a deterministic
version of the well-known random search theorem~\cite{SolisWets81}. Finally, we show and discuss the practical implications of regression to a model described by an \textit{infinite} Taylor series, which is normally outside the scope of existing regression methods.

\subsection{Real number arithmetic}
For examples we focus on the interval $[-1,1]$ which is represented by the type of ternary sequences $\mathbb I=\mathbb N\to \{-1,0,1\}$, a version of the `signed digit representation'~\cite{DBLP:journals/tc/Avizienis61}. Sequences $r:\mathbb I$ are encodings of real numbers in $[-1,1]$, using the standard binary numeral interpretation $\sum_{i\in \mathbb N}r(i)\times 2^{-(i+1)}$.
This representation is particularly well suited for the  definition of  multiplication and normalised addition (taking the midpoint), but is inconvenient for defining an order as the same number can have too many encodings. 
In contrast, $\uniti$ is suitable for ordering but not for arithmetic.
This highlights the convenience of being able to use different representations of the reals for different purposes. 

The midpoint algorithm is closely inspired by Ciaffaglione and Di Gianantonio~\cite{DBLP:journals/tcs/CiaffaglioneG06}, and multiplication by Escard\'o~\cite{Escardo11}. Both of these have been proved formally correct in \emph{loc. cit.} but not in a way that can be easily reused (or recycled) in our setting. 
However, we face an additional burden of proof by being required to show they are all continuous functions in the specific sense of Sec.~\ref{sec:tech}.
This is what we focus on.

Practical applications may require operating with representations of larger sets of reals than just $[-1,1]$.
Arbitrary closed intervals can be obtained from $[-1,1]$ using scaling and shifting by constant values, which introduce some not insurmountable complications.
To deal with larger sets of reals still we need to be always careful that the representation is an \stypet-type.
For instance, a `mantissa and exponent' representation, where the mantissa is a representation of a real and the exponent a natural number is not an \stypet-type.
A good rule of thumb is that compact sets are good candidates which might have such representation.
%An interesting question to ask is, for example, if the the two-point compactification of $\mathbb R$ can be represented as such, along with appropriate arithmetic.
We leave these issues for further work.

The operations below are a minimum set which will allow us to formulate examples. The implementations are meant to be easy to reason about
rather than efficient -- they are in fact not practically usable. To scale up to realistic regression examples as used for example in machine learning the operations need to be much more efficiently implemented and, perhaps, extracted out of \textsc{Agda} into a more performance-oriented language.
However, there is no reason to believe that the recipe we follow below cannot be applied to more, and more efficiently implemented,
operations.

\paragraph{Midpoint} (Details in module \texttt{IAddition})

\newcommand{\zadd}{+}

Let $x\concat x'$ be a sequence with head $x$ and tail $x'$. Let $\mathbb Z$ be the type of integers and $\zadd$ addition  on integers. We write $2i=i\zadd i$. Following \textit{loc. cit.} we define the midpoint operator $\oplus$ using auxiliary  operations
$\lceil-\rceil:\mathbb Z\to\{-1,0,1\}$, $\lfloor-\rfloor : \mathbb Z\to\mathbb Z$ and $  a : U\to U\to \mathbb Z \to U$.
\begin{align*}
\lceil m \rceil &= \begin{cases}
-1 & \text{if $m\leq -2$}\\
0 & \text{if $-2 < m \leq 1$}\\
1 & \text{if $1< m$}\\
\end{cases}\\
\lfloor m \rfloor &= \begin{cases}
m\zadd 4 & \text{if $m\leq -2$}\\
m & \text{if $-2<m\leq 1$}\\
m\zadd(-4) & \text{if $1< m$}
\end{cases}\\
\boxplus(x \concat x', y \concat y', i)_n &= \begin{cases}
\lceil 2 i \zadd x \zadd y\rceil & \text{if $n=0$}\\
\boxplus(x', y', \lfloor 2 i\zadd x \zadd y \rfloor)_m &\text{if $n=m+1$}.
\end{cases}
\end{align*}
The midpoint operator can now be defined as
\[
(x \concat x')\oplus(y \concat y') = \boxplus(x', y', x\zadd y).
\]

 Full blown addition can be defined using $\oplus$ and a global scaling factor via elementary algebraic manipulations. For example, if $u,v,w\in [-1,1]$ then we can define $u+v+w = (u\oplus v)\oplus (w\oplus 0)$ with a global scaling factor of~4.

\begin{lemma}
The $\oplus$ operator is continuous.
\end{lemma}
\begin{proof}
This is so because for  any $n : \mathbb N$, $x_1 \concat x_1'$, $x_2 \concat x_2'$, $y_1 \concat y_1'$, $y_2 \concat y_2' : [-1,1]$ and $z: \mathbb Z$,
if $(x_1 \concat x_1',x_2 \concat x_2')\cong_{(n+1,*), (n+1, *)}(y_1 \concat y_1',y_2 \concat y_2')$ then
$x_1'\oplus x_2' \cong_{(n,*)} y_1'\oplus y_2'$.
\end{proof}
\paragraph{Negation}

In the signed-digit representation this operation is simply reversing the sign of each digit. The continuity is immediate.

\paragraph{Multiplication} (Details in module \texttt{IMultiplication})

Let $\odot$ be multiplication on the set of digits $\{-1,0,1\}$ defined in the obvious way.
And let $\otimes$ be the multiplication of a (code of a) real by a digit, defined by mapping $\odot$ over the sequence of digits which is an element of $\mathbb I$.
We use several auxiliary operations defined by mutual recursion.

First consider $p,p',p'':\mathbb I\times\mathbb I\to \mathbb I$:
\begin{align*}
p(x,y) &= p'(x,y)\oplus p''(x,y)\\
p'(x,y)_n &=\begin{cases}
x' \odot y'' & \text{if $n=0$}\\
(y''\otimes x\oplus x''\otimes y)_n &\text{otherwise}
\end{cases}\\
p''(x,y)&= y'\otimes x \oplus x\otimes y\\
&\text{where }x=x'\concat (x''\concat x'''), y=y'\concat (y''\concat y''').
\end{align*}

The second helper function is $q:\mathbb N\to \mathbb I\times \mathbb I\to \mathbb I$:
\[
q(k,x,y)_n = \begin{cases}
x'\odot y' & \text{if }n=0\\
x''\odot y' & \text{if }n=1\\
x''\odot y'' & \text{if }n=2 \\
0 & \text{if } n>2 \text{ and } k = 0 \\
p(x'',y'')\oplus q(k-1,x'',y'')&\text{otherwise}\\
\end{cases}
\]
where $x=x'\concat (x''\concat x''')$, $y=y'\concat (y''\concat y''')$.

Finally, multiplication $\times:\mathbb I\times\mathbb I\to\mathbb I$ is defined as
\[
(x \times y)_n = \bigl(p(x,y)\oplus q(n,x,y)\bigr)_n.
\]

\begin{lemma}
Multiplication is continuous.
\end{lemma}
\begin{proof}
This amounts to proving that the constituent operators $p$, $p'$, $p''$ and $q$ are continuous. The question is whether at every precision in the exactness type of $\mathbb I$ there exists some MoC in the exactness type of $\mathbb I$. In the cases where the output is a simple arithmetic operation that relies upon zero or one digits of the input -- for example, the $n = 2$ case of $q$ -- the MoC is clear and easily constructed. In all other cases, the output is the result of the composition of the $\oplus$ operator with other operators that have been proved continuous. As $\oplus$ is continuous, it is clear that an MoC can be constructed in these cases too. The most difficult case is the `otherwise' case of $q$, which relies upon constructing an MoC from the continuity of $\oplus$,$p$ and $q$ itself. However, as the $k$ value decreases, we can construct the MoC from an inductive hypothesis on the continuity of $q$ at differing values of $k$.

The formalisation seems forbidding but the intuition is clear. 
\end{proof}

\paragraph{Positive truncation}
The domain of the normalised loss function is $\uniti$, whereas arithmetic happens in $\mathbb I$, for example in computing least-square-like loss functions.
Since we only require continuity and the vanishing property of the loss function, rather than a precise measure of loss, the simplest way to create
a well-typed loss function is to use a `truncation' function $t:\uniti\to\mathbb I$ which changes all digits $-1$ to 0 and keeps the rest of the digits.
This operation preserves continuity and the key property of the loss function, to be vanishing ($f(0)=0$).

Coming back to our discussion in Remark~\ref{rem:real}, the truncation $t$ is a perfect example of a function that operates strictly at the level of \emph{codes} and is not the realiser of a real function $[-1,1]\to[0,1]$.
This is somewhat unsatisfactory from a foundational perspective, but from an algorithmic (and somewhat pragmatic) point of view it raises not serious issues in our setting.
More meaningful loss functions, which are realisers of real functions and have additional desirable properties (e.g. they are monotonic) can be defined, but at the cost of extra complexity.
\\[2ex]
The operations above, together with Prop.~\ref{thm:composition} which states that continuity is preserved by composition, will allow us to construct arbitrary multi-variate polynomial functions. Going beyond that would require extra operations (division, square root, logarithm, trigonometric functions etc.), for which algorithms in real number computation exists, but are beyond the scope of the present work.

\subsection{Search as degenerate regression}

Using the minimisation algorithm for regression (Thm.~\ref{thm:min}) we can compute, to an arbitrary precision, the solution of any continuous function simply by considering the variable(s) as an unknown degenerate model parameter and least squares as the loss function.

Concretely, let us illustrate this with solving a non-linear system of equations:
\begin{align*}
x^2 &= 0.5x\oplus 1\\
y^3 &= x^2.
\end{align*}
This equation is expressed in terms of real numbers, and we can use the minimisation algorithm of Thm.~\ref{thm:min} to look for approximate solutions in $\mathbb I$ with some given precision.
As it happens, the solutions to this equations are both in $[-1,1]$ ($x=-0.851199\ldots$ and $y=0.898161\ldots$).

In the notation of Thm.~\ref{thm:min}, the parameter type $S=\mathbb I\times\mathbb I$, which is an \stypet-type, and $Y=\mathds{1}\to\mathbb I\times\mathbb I\simeq \mathbb I\times\mathbb I$, with $\mathds{1}$ the unit type, which is an \otype-type.
  This is why we call this `degenerate' regression, because the oracle type is not a function type. 

  The model `function' is now a constant:

  \[
M_{(x,y)} = \bigl(
x^2\oplus-(0.5x\oplus 1),
y^3\oplus-x^2
\bigr),
\]

The true (degenerate) oracle is the constant $\Omega=0$.
The loss function is $\Phi : \mathbb I\times\mathbb I\to \mathbb I\times\mathbb I\to\uniti$, $\Phi\bigl(\_,(u,v)\bigr)= t(u\times u\oplus v\times v$).

Since all the types involved are $\mathfrak S$-types and all the functions continuous, as the composition of continuous operations, it is an immediate consequence that the `parameter' $(x,y):\mathbb I\times\mathbb I$ can be computed for whatever precision $p:\mathbb N$.
The minimiser used in the theorem is a possible such algorithm.

Two caveats are required. The first one is that regression will compute the `\textit{argmin}' of the function, so it will return one of the solutions if they exist.
This has been already discussed in the general setting in Remark~\ref{rem:real}.
In this example both real solutions are in $[-1,1]$.
The algorithm does not control which one will be returned.
The second one is that in the case of no solution the minimisation algorithm will still return some $(x_0,y_0)$ value of the \textit{argmin}, so the model itself must be
used to test whether the loss value is close enough to zero to be considered a solution.
Whether a returned pair is an actual solution, i.e. $\Phi(\Omega,M_{(x_0,y_0)}) = 0$ is, of course, not decidable because
equality is not decidable in $\uniti$.

\subsection{Polynomial regression}
\label{sec:polyreg}

This is the `meat and potatoes' motivating example. Consider a set of points $(x_i,y_i) \in \mathbb R^2$, $i\leq n$. And suppose that we want to `best fit' a polynomial $f_{\vec k}(x)=k_0+k_1x+\cdots k_mx^m : \mathbb R \to \mathbb R$ through this data set, i.e. find values for  $\vec k\in \mathbb R^{m+1}$  which minimise a loss function such as least squares.

One apparent obstacle is that all the convergence theorems require an \textit{oracle} to compute the parameter $\vec k$, whereas we only have a set of points. An important observation is that the least square loss function computes the loss only at the given data-points and ignores its behaviour elsewhere. So any `oracle' constructed from the points which is continuous would ultimately lead to the same result.

To construct such an oracle we can use interpolation. There are many interpolation algorithms but for our purpose we might as well take the simplest one: piece-wise constant interpolation. 
% (TODO: Think this is no longer needed?):
% Since the data sample is contained in the interval, we also use constant \textit{extrapolation} to defined the oracle between the extremal data points and the ends of the interval. 

Let $p$ be some fixed precision $p:\mathbb N$ and $y_{n}:\mathbb I$ an arbitrary value.
We define a (distorted) oracle from the data points ($x_i, y_i:\mathbb I$):

\[
\Psi_\Omega(x)=\begin{cases}
y_0 & \text{if } x <_p x_0\oplus x_1\\
y_1 & \text{if } x <_p x_1\oplus x_2\\
\vdots\\
y_{n-1} & \text{if } x <_p x_{n-1}\oplus x_n \\
y_n & \text{otherwise.}
\end{cases}
\]

The definition assumes that the data points are sorted by the $x_i$ component.
This function is defined by cases noting that the order in which the conditions are tested is fixed, top-to-bottom. This makes the function well defined, computable, and, perhaps surprisingly, continuous.
(In fact Thm.~\ref{thm:imperfect} does not require the oracle to be continuous, just Thm.~\ref{thm:congstyle}.)
The real issue is not continuity but why the function is well defined.
The function is defined piecemeal, but if $x$ is closer to some $x_i\oplus x_{i+1}$ than the precision $p$ then we cannot say for sure if it is to the left or to the right of it.
In this situation the fact that the side-conditions are checked in a defined order means $x$ will be considered as if it is to the left of the $x_i$, which makes the function well defined. 
Note that this means that this function is also not a realiser of a continuous function $\mathbb R\to \mathbb R$, an issue which we discussed before (Remark~\ref{rem:real}).

The general property of regression (Thm.~\ref{thm:min}) guarantees that parameters $\vec k$ can be computed so that the interpolated model $M_{\vec k}$ will minimise the least-square error at each $x_i$. It is interesting to also consider what this means from the point of view of convergence. The perfect-model convergence theorem (Thm.~\ref{thm:perfect}) is not applicable since the general form of the oracle (line segments) and of the model (polynomial) are not the same.

However, the general imperfect-model convergence theorem (Thm.~\ref{thm:imperfect}) says that if the loss between a distorted model and the true model vanishes then so does the loss between the true model and the regressed model. In this case the true oracle would be a same (or lower) degree polynomial from which the data points are sampled then interpolated, resulting in the distorted oracle. Since the least squares loss function only considers the behaviour at the sample points, it will be zero when applied to the true and distorted oracles. Which means that Thm.~\ref{thm:imperfect} guarantees that in this situation the loss between the true oracle and the model can be also made arbitrarily small.
From which we can conclude that polynomial regression, as performed in practice, has good convergence properties.

The possibly problematic aspect of this is not the use of a polynomial as a model but the fact that we are working `offline' (from data) as opposed to `online' (from the oracle). 
But the correctness of `offline' regression is an immediate corollary of Thm.~\ref{thm:imperfect}

\begin{proposition}[\texttt{Examples.offline-regression}]\label{prop:offline}
    Let  $S$ be an $\mathfrak S$-type, $Y=\mathbb I\to\mathbb I$, $p:\mathbb N$ a precision, $\epsilon :\uniti$ a loss value, points $\vec x:\mathbb I^n$, $n:\mathbb N$, $\Phi_{\vec x}$ a least-squares loss function, and $\Psi_{\vec x}:Y\to Y$ a constant interpolation function, both defined at points $\vec x$. 

  There exists a regressor $\mathrm {reg}:(S\to Y)\to Y\to S$ such that given a weakly continuous model $M:S\to Y$ 
  for parameter $k=\mathrm{reg}\,M\,\Psi_{\vec x}(\Omega)$:
  \[
  \text{if }\Phi_{\vec x}(\Psi_{\vec x}(\Omega), \Omega)<_p\epsilon \text{ then } \Phi_{\vec x}(\Psi_{\vec x}(\Omega), M_k)<_p\epsilon,
  \]
  for synthetic oracle $\Omega = M(k_0)$.
\end{proposition}

From this, the convergence of polynomial interpolation follows immediately as any model defined by a polynomial is continuous. 

\subsection{Universal approximators}

In applications, particularly to machine learning, we may not know the general form of the oracle.
In such a situation we may want to consider a more general kind of model, which is expressed as an \emph{infinite} series, such as power series or trigonometric series. 
Many such series can be written in the form $M_k(x)=\sum_{i\in \mathbb N} f(k_n, x, n)$ where $f$ is a fixed function and $k_n$ an infinite set of parameters.
For example, in the case of a power series $f(k, x, n) = kx^n$.
Such series can serve as `universal approximators' for classes of functions.
For example, analytic functions equal to their Taylor series at all points form a class known as `integral' functions. 
The polynomials, exponential, and certain trigonometric (sine and cosine) functions are examples of integral functions.

These models are intriguing because they can be given types such as $M:(\mathbb N\to \uniti)\to \uniti\to \uniti$, with the type of parameters $k:\mathbb N\to\uniti$ an \stypet-type.
This means that, providing the continuity of $M$ is proved, the \emph{entire} set of parameters $k$ can be computed to any degree of precision.
In the case of the power series we know that using addition, multiplication, and composition always leads to continuous functions.
The problem is computing the infinite series.
Provided that the series converges, the series can be computed in general~\cite{muller1995constructive} or approximated~\cite{devore1993constructive}, but this is beyond the scope of our paper. 
From the point of view of regression analysis this may seem surprising, but this is a known result using searchable sets~\cite{DBLP:journals/logcom/Escardo13}.

For example consider a model $M_k(x)=\bigoplus_{i:\mathbb N}k_ix^i$, which converges for all values of $x:\uniti$.
It can be used to regress some oracle $\Omega:\uniti\to\uniti$.
Using the regressor of Thm.~\ref{thm:perfect}, the sequence of parameters is given by $k_0 = \mathrm{reg}\,M\,\Omega$, so a model can be instantiated as $M_{k_0}$.

Note that the solution above involves an infinite set of parameters so obviously cannot be computed other than lazily.
The model, after instantiating $k$ is
\[
M(x)=\bigoplus_{i:\mathbb N}(\mathrm{reg}\,M\,\Omega)(i)\cdot x^i,
\]
which is computable but could be expensive to compute.

A problem of practical importance in this setting is the `truncation' of the series defining model $M_k$ to only a finite number of terms, i.e. $M_{k,m}=\bigoplus_{0\leq i\leq m}k_ix^i$.
However, such a model has type $M:(\mathbb N\to \uniti)\to \mathbb N \to \uniti\to \uniti$.
The type $\mathbb N$ is not an \stypet-type; it is also clearly not a searchable type.
So this problem cannot be solved.

A broader consequence is that some `hyper-parameters' of neural networks (the number of layers, the number of neurons per layer, etc.) also cannot be computed using our approach.

\begin{remark}
  This class of more speculative examples, in particular summing infinite series, is not formalised in \textsc{Agda}.
\end{remark}

\section{Related work}

This paper has been inspired by and relies extensively on a significant body of work by Escard\'o, starting with \emph{searchable infinite sets}~\cite{DBLP:conf/lics/Escardo07}.
The properties of regression established here can be equally formulated in that setting, or in related setting such as \emph{compact sets}~\cite{DBLP:journals/logcom/Escardo13} and \emph{compact types}~\cite{CompactTypes}.
What makes our approach distinct is that in the formulations above are \emph{not} synthetic, in the sense of~\cite{DBLP:journals/entcs/Escardo04}.
Whereas in synthetic topology all functions are assumed to be continuous, we work with an explicit condition of continuity.
This makes proofs more difficult but it has the advantage that makes our regression theorems hold in more models of type theory, including those that manipulate non-continuous functions, yet allowing for formalisation in a proof assistant based on dependent type theory (namely \textsc{Agda}). 

We are interested in establishing an alternative framework for a better mathematical understanding of data science, machine learning, etc. based on type theory and constructive real numbers. 
It is worth drawing an anaology it with the established mathematical framework for machine learning, \emph{probably approximately correct learning} (PAC)~\cite{DBLP:journals/cacm/Valiant84,DBLP:conf/aaai/Haussler90}.
We first introduce its basic concepts.

Let $X$ be a set and $f:X\to \{0,1\}$ an unknown function (in our terminology, the `oracle'). 
A sample $\vec x$ is drawn from $X$ according to some (unknown) distribution $\mathcal D$ and is correctly classified according to $f$.
Can we learn the function $f$?
Note that this is a particular instance of regression as discussed here.

The function $f$ is not usually guessed out of nothing, but from a known class of possible functions $\mathcal H$, dubbed \emph{inductive bias}.
Our counterpart is, of course, the class of \emph{models} $M$.
The working assumption is that $f\in\mathcal H$, which is mirrored in our approach, in the convergence theorems, by the fact that $\Omega=M_{k_0}$ for some unknown~$k_0$.

Suppose that a \emph{learning procedure} (which we call a `regressor') produces a new hypothesis $h_{\vec x}\in \mathcal H$ based on the sample.
This is what we call a `regressed model' $M_k$.
The basic question is how good is this new hypothesis?
It should be good for the sample, but also for new examples.

The error is defined as $\mathrm{err}(h_{\vec x})=\mathbf{Pr}_{x\sim\mathcal D}[h_{\vec x}(x)\neq f(x)]$, the probability that under the given distributions the unknown function and the hypothesis differ.
The problem statement is that given an error $\epsilon\in(0,1)$ what can it be said about $\mathrm{err}(h_{\vec x})\leq \epsilon$.
This cannot be guaranteed, except with probability at least $1-\delta$ for some fixed parameter $\delta\in(0,1)$. 

A hypothesis class $\mathcal H$ is \emph{PAC-learnable} if there is an algorithm such that for every $\epsilon,\delta\in(0,1)$ and unknown function $f \in\mathcal  H$ there is a natural number $m > 0$ such that when running the algorithm  on an sample of independent examples $x_i, i=0,m$ according to some distribution $\mathcal D$, we obtain an $h_{\vec x}\in\mathcal H$ such that $\mathrm{err}(h_{\vec x}) \leq\epsilon$ with probability at least $1-\delta$.

The size of the sample $m$ given as a function of $\delta^{-1}$ and $\epsilon^{-1}$ is called the \emph{sample complexity}.
Finite sets are a non-surprising example of PAC-learnable and their sample complexity bounds are known.
But certain infinite sets also are PAC-learnable, with sample complexity determined by the so-called \emph{Vapnik–Chervonenkis  dimension} (VC)~\cite{doi:10.1137/1116025}. 

Our approach is complementary to PAC, having certain strengths and weakness (leaving aside the obvious fact that PAC theory is a mature and well explored area of research).
The setting of the problem is similar, up to differences in vocabulary, but both the learning procedure and the validation procedure vary significantly.
PAC requires a prior sampling of the oracle with a given distribution, which makes it intrinsically `off-line', whereas our learning procedure assumes access to the oracle, `on-line'.
(The two are related by Prop.~\ref{prop:offline}, but more about this in the next section.)
It also means that the learning procedure and the testing criterion in PAC are necessarily probabilistic.
In contrast, our approach is deterministic and quantitative in a different way: instead of measuring the \emph{probability} of the learned outcome being different from the desired outcome we measure the \emph{definite amount} by which the two outcomes differ.
For finite sets, which can be searched trivially, our approach is trivial whereas the PAC is interesting.
But for infinite sets both our approach and the PAC approach give interesting and non-trivial characterisations. 

\section{Conclusions and further work}

The main contribution of the paper is to offer a range of convergence criteria for parametric regression, formalised in type theory, and proved formally in \textsc{Agda}.
The main convergence theorem (Thm.~\ref{thm:imperfect}) states that a large class of oracles, all continuous functions of \otype-type with unknown parameters of \stypet-type, can be regressed up to any desired precision, even in the presence of distortions, so long as the distortions are small.
The regressors used in the theorem can be considered as correct, albeit inefficient, reference implementations that satisfy the conditions of the convergence theorem.

The next part of this work will requires us to turn our attention to off-line learning.
The starting point is Prop.~\ref{prop:offline} which gives a convergence criterion for off-line regression.
The interesting part is the precondition $\Phi_{\vec x}(\Psi_{\vec x}(\Omega), \Omega)<_p\epsilon$.
We conjecture that there if the sample $\vec x$ is large enough then this precondition is always true.
The reason is that the distorted oracle $\Psi_{\vec x}(\Omega)$ constructed by interpolation should become close enough to the true oracle $\Omega$ as the sample grows, which is a version of the Stone-Weierstrass interpolation theorem.
Our simple interpolator (piece-wise constant) may not be suitable for such a theorem, but we strongly believe that such interpolators exist in our setting.
Interpolation, as mentioned above, is closely related to sampling, which could open the door to dealing with probabilistic sampling and formulating convergence results more closely related to PAC learning, including estimating or bounding the sample size.
The fact that probabilities over discrete sets such as $\mathbb N\to[0,1]$ are \stypet-types is encouraging. 
In the longer term we also wish to find (synthetic) topological or type-theoretic characterisation of other PAC concepts, such as VC dimension.

Interpolation in itself is very important because, especially in the presence of distortions (noise), as it forms the basis of \emph{non-parametric regression}, the learning of models without committing to a particular shape of a model.

A better class of interpolation functions should also resolve the foundational rough edges discussed in Sec.~\ref{sec:polyreg}, namely the fact that the interpolated functions are not realisers of real functions.
We do not believe these issues have any profound consequences but are best avoided.
In contrast, the same issues in the context of the minimisation theorem (Thm.~\ref{thm:min}) cannot be solved --- but this theorem is a `dead end' for us.

In parallel we aim to consider more realistic implementations, either extracted from the \textsc{Agda} regressors or implemented directly in other more performance-oriented languages.
The key requirement is fast (enough) arbitrary precision arithmetic over real numbers, a field intensely studied with multiple libraries available for various languages~\cite{plume,DBLP:journals/jlp/Menissier-Morain05,DBLP:journals/tcs/Briggs06}.

%%
%% The next two lines define the bibliography style to be used, and
%% the bibliography file.
\bibliographystyle{abbrv}
\bibliography{arxiv}

\begin{thebibliography}{10}

\bibitem{DBLP:journals/tc/Avizienis61}
A.~Avizienis.
\newblock Signed-digit numbe representations for fast parallel arithmetic.
\newblock {\em {IRE} Trans. Electronic Computers}, 10(3):389--400, 1961.

\bibitem{DBLP:conf/tphol/BoveDN09}
A.~Bove, P.~Dybjer, and U.~Norell.
\newblock A brief overview of agda - {A} functional language with dependent
  types.
\newblock In S.~Berghofer, T.~Nipkow, C.~Urban, and M.~Wenzel, editors, {\em
  Theorem Proving in Higher Order Logics, 22nd International Conference, TPHOLs
  2009, Munich, Germany, August 17-20, 2009. Proceedings}, volume 5674 of {\em
  Lecture Notes in Computer Science}, pages 73--78. Springer, 2009.

\bibitem{DBLP:journals/tcs/Briggs06}
K.~Briggs.
\newblock Implementing exact real arithmetic in python, {C++} and {C}.
\newblock {\em Theor. Comput. Sci.}, 351(1):74--81, 2006.

\bibitem{DBLP:journals/tcs/CiaffaglioneG06}
A.~Ciaffaglione and P.~D. Gianantonio.
\newblock A certified, corecursive implementation of exact real numbers.
\newblock {\em Theor. Comput. Sci.}, 351(1):39--51, 2006.

\bibitem{devore1993constructive}
R.~A. DeVore and G.~G. Lorentz.
\newblock {\em Constructive approximation}, volume 303.
\newblock Springer Science \& Business Media, 1993.

\bibitem{DBLP:journals/entcs/Escardo04}
M.~H. Escard{\'{o}}.
\newblock Synthetic topology: of data types and classical spaces.
\newblock {\em Electr. Notes Theor. Comput. Sci.}, 87:21--156, 2004.

\bibitem{DBLP:conf/lics/Escardo07}
M.~H. Escard{\'{o}}.
\newblock Infinite sets that admit fast exhaustive search.
\newblock In {\em 22nd {IEEE} Symposium on Logic in Computer Science {(LICS}
  2007), 10-12 July 2007, Wroclaw, Poland, Proceedings}, pages 443--452. {IEEE}
  Computer Society, 2007.

\bibitem{DBLP:journals/lmcs/Escardo08}
M.~H. Escard{\'{o}}.
\newblock Exhaustible sets in higher-type computation.
\newblock {\em Logical Methods in Computer Science}, 4(3), 2008.

\bibitem{Escardo11}
M.~H. Escard{\'{o}}.
\newblock Real number computatio in {H}askell with real numbers represented as
  infinite sequences of digits.
\newblock \url{https://www.cs.bham.ac.uk/~mhe/papers/fun2011.lhs}, 2011.

\bibitem{DBLP:journals/logcom/Escardo13}
M.~H. Escard{\'{o}}.
\newblock Algorithmic solution of higher type equations.
\newblock {\em J. Log. Comput.}, 23(4):839--854, 2013.

\bibitem{CompactTypes}
M.~H. Escard\'o.
\newblock Compact types.
\newblock \url{http://www.cs.bham.ac.uk/~mhe/agda-new/CompactTypes.html}, 2018.

\bibitem{floudas2009review}
C.~A. Floudas and C.~E. Gounaris.
\newblock A review of recent advances in global optimization.
\newblock {\em Journal of Global Optimization}, 45(1):3, 2009.

\bibitem{DBLP:conf/aaai/Haussler90}
D.~Haussler.
\newblock Probably approximately correct learning.
\newblock In H.~E. Shrobe, T.~G. Dietterich, and W.~R. Swartout, editors, {\em
  Proceedings of the 8th National Conference on Artificial Intelligence.
  Boston, Massachusetts, USA, July 29 - August 3, 1990, 2 Volumes}, pages
  1101--1108. {AAAI} Press / The {MIT} Press, 1990.

\bibitem{lawler1966branch}
E.~L. Lawler and D.~E. Wood.
\newblock Branch-and-bound methods: A survey.
\newblock {\em Operations research}, 14(4):699--719, 1966.

\bibitem{DBLP:journals/jlp/Menissier-Morain05}
V.~M{\'{e}}nissier{-}Morain.
\newblock Arbitrary precision real arithmetic: design and algorithms.
\newblock {\em J. Log. Algebr. Program.}, 64(1):13--39, 2005.

\bibitem{mooreyoung59}
R.~Moore and C.~Yang.
\newblock Interval analysis.
\newblock Space Division Report LMSD285875, Lockheed Missiles and Space Co,
  1959.

\bibitem{muller1995constructive}
N.~T. M{\"u}ller.
\newblock Constructive aspects of analytic functions.
\newblock In {\em Proceedings of Workshop on Computability and Complexity in
  Analysis}, volume 190, pages 105--114. Informatik Berichte
  FernUniversit{\"a}t Hagen, 1995.

\bibitem{neumaier2004complete}
A.~Neumaier.
\newblock Complete search in continuous global optimization and constraint
  satisfaction.
\newblock {\em Acta numerica}, 13:271--369, 2004.

\bibitem{pearl2018build}
J.~Pearl.
\newblock To build truly intelligent machines, teach them cause and effect.
\newblock {\em Quanta Magazine (15 May 2018)}, 2018.

\bibitem{PINKUS20001}
A.~Pinkus.
\newblock Weierstrass and approximation theory.
\newblock {\em Journal of Approximation Theory}, 107(1):1 -- 66, 2000.

\bibitem{piyavskii1972algorithm}
S.~Piyavskii.
\newblock An algorithm for finding the absolute extremum of a function.
\newblock {\em USSR Computational Mathematics and Mathematical Physics},
  12(4):57--67, 1972.

\bibitem{plume}
D.~Plume.
\newblock A calculator for exact real number computation, 1998.
\newblock University of Edinburgh.

\bibitem{ruder2016overview}
S.~Ruder.
\newblock An overview of gradient descent optimization algorithms.
\newblock {\em arXiv preprint arXiv:1609.04747}, 2016.

\bibitem{DBLP:conf/mfcs/Simpson98}
A.~K. Simpson.
\newblock Lazy functional algorithms for exact real functionals.
\newblock In {\em Mathematical Foundations of Computer Science 1998, 23rd
  International Symposium, MFCS'98, Brno, Czech Republic, August 24-28, 1998,
  Proceedings}, pages 456--464, 1998.

\bibitem{skelboe1974computation}
S.~Skelboe.
\newblock Computation of rational interval functions.
\newblock {\em BIT Numerical Mathematics}, 14(1):87--95, 1974.

\bibitem{SolisWets81}
F.~J. Solis and R.~J.~B. Wets.
\newblock Minimization by random search techniques.
\newblock {\em Math. Oper. Res.}, 6(1):19--30, Feb. 1981.

\bibitem{tawarmalani2001semidefinite}
M.~Tawarmalani and N.~V. Sahinidis.
\newblock Semidefinite relaxations of fractional programs via novel
  convexification techniques.
\newblock {\em Journal of Global Optimization}, 20(2):133--154, 2001.

\bibitem{Troelstra}
A.~Troelstra and D.~van Dalen.
\newblock Chapter 6 some elementary analysis.
\newblock In {\em Constructivism in Mathematics}, volume 121 of {\em Studies in
  Logic and the Foundations of Mathematics}, pages 291 -- 325. Elsevier, 1988.

\bibitem{DBLP:journals/cacm/Valiant84}
L.~G. Valiant.
\newblock A theory of the learnable.
\newblock {\em Commun. {ACM}}, 27(11):1134--1142, 1984.

\bibitem{doi:10.1137/1116025}
V.~N. Vapnik and A.~Y. Chervonenkis.
\newblock On the uniform convergence of relative frequencies of events to their
  probabilities.
\newblock {\em Theory of Probability \& Its Applications}, 16(2):264--280,
  1971.

\end{thebibliography}

%%
%% If your work has an appendix, this is the place to put it.
%\appendix

%\section{Extra stuff}

\end{document}